\documentclass[10pt,conference]{IEEEtran} 
\usepackage{makeidx}

\usepackage[utf8]{inputenc}
\usepackage[T1]{fontenc}

\usepackage[english]{babel}

\usepackage{url}
\newcommand{\hairspace}{\hspace{1pt}}
\newcommand{\eg}{\mbox{e.\hairspace{}g}.\ }
\newcommand{\ie}{\mbox{i.\hairspace{}e}.\ }

\newcommand{\etal}{\mbox{et~al.}\ }
\usepackage[cmex10]{amsmath}
\usepackage{amssymb}
\usepackage{amsthm}

\newtheorem{definition}{Definition}
\newtheorem{theorem}{Theorem}
\newtheorem{lemma}{Lemma}

\usepackage{bbding}
\usepackage{tikz}
\usepackage{listings}
\usepackage{cite}
\usepackage{MnSymbol}
\usepackage[binary-units=true]{siunitx}
\mathchardef\mhyphen="2D 

\usepackage{booktabs}
\usepackage{array}
\usepackage{multirow}

\usepackage{pifont}

\usepackage{expdlist}

\usepackage{fancyvrb}

\usepackage[pdftex,hyperfootnotes=false,pdfpagelabels,bookmarks,linktoc=section]{hyperref}

\hypersetup{
    colorlinks=true,
    linktocpage=true,
    breaklinks=true,
    bookmarksnumbered,
    bookmarksopen=true,
    bookmarksopenlevel=1,
    pdfhighlight=/O,
    pdfpagemode=UseOutlines,
    urlcolor=black,
    linkcolor=black,
    citecolor=black,
    pdftitle={Certifying Spoofing-Protection of Firewalls},
    pdfauthor={Cornelius Diekmann, Lukas Schwaighofer, and Georg Carle}
}
\newcommand\lstapp{\ensuremath{\mathbin{:\mkern-1mu:\mkern-1mu:}}} 
\newcommand\lstcons{\ensuremath{\mathbin{:\mkern-1mu:}}} 

\usepackage{paralist}

\begin{document}

\title{Certifying Spoofing-Protection of Firewalls}

\author{\IEEEauthorblockN{Cornelius Diekmann, Lukas Schwaighofer, and Georg Carle}
\IEEEauthorblockA{Technische Universit{\"a}t M{\"u}nchen\\%
Email: $\lbrace \textnormal{diekmann} \vert \textnormal{schwaighofer} \vert \textnormal{carle} \rbrace \textnormal{@net.in.tum.de}$}
}

\maketitle

\thispagestyle{empty}

\begin{abstract}
We present an algorithm to certify IP spoofing protection of firewall rulesets.
The algorithm is machine-verifiably proven sound and its use is demonstrated in real-world scenarios.
\end{abstract}

\section{Introduction}
In firewalls, it is good practice and sometimes an essential security feature to prevent IP address spoofing attacks \cite{rfc2827,rfc3704,tolarisblogdisablerpfilter,bartal1999firmato,marmorstein2005itval,wool2004use}.
The Linux kernel offers reverse path filtering~\cite{kernelrpfilter}, which can conveniently prevent such attacks. 
However, in many scenarios (\eg asymmetric routing~\cite{iptablesperfectruleset}, failover~\cite{tolarisblogdisablerpfilter}, zone-spanning interfaces, or multihoming~\cite{rfc3704}), reverse path filtering must be disabled or is too coarse-grained (\eg for filtering special purpose addresses\cite{rfc6890}).
In these cases, spoofing protection has to be provided by the firewall which is configured by the administrator.
Unfortunately, writing firewall rules is prone to human error, a fact known for over a decade~\cite{firwallerr2004}.

For Linux netfilter/iptables firewalls~\cite{iptables}, we present an algorithm to certify spoofing protection of a ruleset.
It provides the following contributions; a unique set of features in their combination:
\begin{itemize}
	\item It is formally and machine-verifiably proven sound with Isabelle/HOL.
	\item It supports the largest subset of iptables features compared to any other firewall analysis systems.
	\item It was tested on the largest (publicly-available) firewall which was ever analyzed in academia.
	\item It terminates within a second for thousands of rules.
\end{itemize}

We chose iptables because it provides one of the most complex firewall semantics widely deployed~\cite{fireman2006,pozo2009model}.
Of course, our algorithm is also applicable to similar or less complex (\eg Cisco PIX) firewall systems.

\section{Related Work}
There are several popular static firewall analysis tools.
The Firewall Policy Advisor~\cite{alshaer2004discoverypolicanomalies} discovers inconsistencies between pairs of rules (\eg one rule completely overshadowing the other) in distributed firewall setups.
A similar tool is FIREMAN~\cite{fireman2006}. 
It can discover inconsistencies within rulesets or between firewalls and verify setups against administrator-defined policies.
To represent sets of packets, Binary Decision Diagrams (BDDs) are used.
It does not support matching on interfaces in a rule.
Since interfaces can be strings of arbitrary length, they may not be ideal for the encoding in BDDs and adding support to FIREMAN might be complicated or deteriorate its performance.

Margrave~\cite{nelson2010margrave} can be used to query (distributed) firewalls.
It is well suited to troubleshoot and to debug firewall configurations or to show the impact of ruleset edits.
For certain queries, it can also find scenarios, \eg it can show concrete packets which violate a security policy.
This scenario finding is sound but not complete.
Its query language as well as Margrave's internal implementation is based on first-order logic. 
A similar tool (relying on BDDs), with a different query language and focus on complete networks, is ConfigChecker~\cite{alshaer2011configcheckershort,alshaer2009configchecker}.
ITVal~\cite{marmorstein2005itval} can also be used to query firewalls.
It can also describe the firewall in terms of equivalent IP address spaces~\cite{marmorstein2006firewall}, which does not require the administrator to pose specific queries.

None of these tools can directly verify spoofing protection.
Nor are the tools themselves formally verified, which limits confidence in their results. 
In addition, the tools only support a limited subset of real-world firewalls~\cite{diekmann2015fm}.
If the firewall under analysis exceeds this feature set, the tools either produce erroneous results or cannot continue~\cite{diekmann2015fm}.

Jeffrey and Samak~\cite{fwmodelchecking2009jeffry} analyze the complexity of firewall analysis and conclude that most questions (for variable packet models) are NP-complete.
They show that SAT solvers usually outperform BDD-based implementations.

\section{Mathematical Background}
This work was done completely in the Isabelle theorem prover~\cite{isabelle2015,diekmanngithubiptablessemantics}.
Isabelle is an LCF-style theorem prover.
This means, a fact can only be proven if it is accepted by a mathematical inference kernel. 
This kernel is very small and thoroughly inspected by the formal methods community.
This makes errors very unlikely, which has been demonstrated by Isabelle's success over the past 20 years.
Our Isabelle formalization is publicly available~\cite{diekmanngithubiptablessemantics}.
An interested reader can replay the proofs and results of the evaluation on her system.
For brevity, we only outline a proof's core idea in this paper and refer the interested reader to our proof document for for further mathematical details.

Our notation is close to Isabelle, Standard ML, or Haskell:
Function application is written without parentheses, \eg $f\ a$ denotes function $f$ applied to parameter $a$. 
For lists, we denote $\mathit{cons}$ and $\mathit{append}$ by `$\lstcons$' and `$\lstapp$', \eg `$x \lstcons \left[y,z\right] \lstapp [a]$'.
Linux shell commands are set in \texttt{typewriter} font.

\paragraph*{Iptables Semantics}
To define and prove correctness of our algorithm, one must first define the semantics (\ie behavior) of an iptables firewall.
We rely on the semantics specified by Diekmann \etal\cite{diekmann2015fm}.
The semantics defines the following common actions (also called ``targets'' in iptables terminology):
$\mathtt{Accept}$, $\mathtt{Drop}$, $\mathtt{Reject}$, $\mathtt{Log}$, calling to and $\mathtt{Return}$ing from user-defined chains, as well as the ``empty'' action.
Matching a packet against match conditions (\eg IP source or destination addresses) is mathematically defined with a ``magic oracle'' which understands \emph{all} possible matches.

Obviously, semantics with a ``magic oracle'' cannot be expressed in terms of executable code.
Nevertheless, the algorithm we present in this paper is both executable and proven sound w.\,r.\,t.\ these semantics.

\section{Spoofing Protection -- Mathematically}
To define spoofing protection, two data sets are required:
The firewall ruleset $\mathit{rs}$ and the IP addresses assignment $\mathit{ipassmt}$.
$\mathit{ipassmt}$ is a mapping from interfaces to an IP address range.
Usually, it can be obtained by \texttt{ip route}.
We will write $\mathit{ipassmt}[i]$ to get the corresponding IP range of interface $i$.
For the following examples, we assume 
\begin{IEEEeqnarray*}{l}
\mathit{ipassmt} = [\mathtt{eth0} \mapsto \lbrace 192.168.0.0/24 \rbrace ]
\end{IEEEeqnarray*}

\begin{definition}[Spoofing Protection]
\label{def:nospoofstrict}
\textnormal{
A firewall ruleset provides \emph{spoofing protection} if for all interfaces $i$ specified in $\mathit{ipassmt}$, all packets from $i$ which are accepted by the ruleset have a source IP address contained in $\mathit{ipassmt}[i]$.
}\end{definition}

For example, using pseudo iptables syntax, the following ruleset implements spoofing protection:
\begin{Verbatim}[frame=single,label={firewall},labelposition=bottomline]
-i eth0 --src !192.168.0.0/24 Drop
any Accept
\end{Verbatim}


\paragraph*{Spoofing Protection with Unknowns}
iptables supports numerous match conditions and new ones may be added in future.
It is practically infeasible to support all match conditions in a tool~\cite{diekmann2015fm}.
As we don't want our algorithm to abort if an unknown match occurs, we will refine Def.~\ref{def:nospoofstrict} to incorporate unknown matches.
This is motivated by the following examples.

We assume that \verb~--foo~ is a match condition which is unknown to us.
Therefore, we cannot guarantee that
\begin{Verbatim}[frame=single,label={firewall},labelposition=bottomline]
-i eth0 --src !192.168.0.0/24 --foo Drop
any Accept
\end{Verbatim}
implements spoofing protections since \verb~--foo~ could prevent certain spoofed packets from being dropped.
Also, the following ruleset might neither implement spoofing protection since ~\verb~--foo~ might allow spoofed packets:
\begin{Verbatim}[frame=single,label={firewall},labelposition=bottomline]
--foo Allow
-i eth0 --src !192.168.0.0/24 Drop
any Accept
\end{Verbatim}
However, the following ruleset definitely implements spoofing protection; 
Independently of the meaning of \verb~--foo~ and \verb~--bar~, it is guaranteed that no spoofed packets are allowed:
\begin{Verbatim}[frame=single,label={firewall},labelposition=bottomline]
--foo Drop
-i eth0 --src !192.168.0.0/24 Drop
--bar Accept
\end{Verbatim}

\noindent
This motivates Def.~\ref{def:nospoofpotentially}.

\begin{definition}[Certifiable Spoofing Protection]
\label{def:nospoofpotentially}
\textnormal{
A firewall ruleset provides \emph{spoofing protection} if for all interfaces $i$ specified in $\mathit{ipassmt}$, all packets from $i$ which are \emph{potentially} accepted by the ruleset have a source IP address in the IP range of $i$.
}\end{definition}

This new definition is stricter than the original one: \mbox{Def.\,\ref{def:nospoofpotentially}} implies \mbox{Def.\,\ref{def:nospoofstrict}}.\footnote{Formally proven; It follows from~\cite[Thm\,3]{diekmann2015fm}}
Therefore the new definition is \emph{sound} and can be used to prove that the last example implements spoofing protection.

However, depending on the meaning of \verb~--foo~, some of the previous examples might also implement spoofing protection.
This cannot be shown with Def.\,\ref{def:nospoofpotentially}, so the new definition is not \emph{complete}.
However, as long as we anticipate unknown matches to occur, it is impossible to obtain completeness.

\section{Spoofing Protection -- Executable}
A straight forward spoofing protection proof of a firewall ruleset using Def.~\ref{def:nospoofpotentially} would require iterating over all packets, which is obviously infeasible.
We present an efficient executable algorithm to certify spoofing protection.

We assume the ruleset to be certified is preprocessed.
For this, we rely on the semantics-preserving ruleset simplification~\cite{diekmann2015fm}, which rewrites a ruleset to a semantically equivalent ruleset where only $\mathit{Accept}$ and $\mathit{Drop}$ actions occur.\footnote{The correctness of this preprocessing is also verified with Isabelle.}
A preprocessed ruleset always has an explicit deny-all or allow-all rule at the end; it can never be empty.
This rule corresponds to a chain's default policy.

We call our algorithm \texttt{sp} (``spoofing protection''). 
It certifies spoofing protection for one interface $i$ in $\mathit{ipassmt}$.
Using \texttt{sp} to certify all $i \in \mathit{ipassmt}$ for a ruleset implies spoofing protection according to Def.~\ref{def:nospoofpotentially}.

We assume the global, static, and fixed parameters of \texttt{sp} are an interface $i$ and the $\mathit{ipassmt}$. 
Then, \texttt{sp} has the following type signature: 
\begin{IEEEeqnarray*}{c}
\mathit{rule}\ \mathit{list} \ \times \ \mathit{ipaddr}\ \mathit{set} \ \times \ \mathit{ipaddr}\ \mathit{set} \to \mathbb{B}
\end{IEEEeqnarray*}
The first parameter ($\mathit{rule}\ \mathit{list}$) is the preprocessed firewall ruleset.
A rule is a tuple $(m, a)$, where $m$ is the match condition and $a$ the action.
The action is either $\mathit{Accept}$ or $\mathit{Drop}$.
For a packet $p$, there is a predicate $\mathit{matches} \ m \ p$ which tells whether the packet $p$ matches the match condition $m$. 

The second parameter ($\mathit{ipaddr}\ \mathit{set}$) is the set of potentially allowed source IP addresses for $i$. 
The third parameter ($\mathit{ipaddr}\ \mathit{set}$) is the set of definitely denied source IP addresses for $i$. 

The last parameter ($\mathbb{B}$) is a Boolean, which is true if spoofing protection could be certified.

Before we present the algorithm, we first present its correctness theorem (which will be proven later).

\begin{theorem}[\texttt{sp} sound]
\label{thm:spsound}
\addpenalty{-10000}
\textnormal{For any ruleset $\mathit{rs}$, if }
\begin{IEEEeqnarray*}{lcl}
\forall i \in \mathit{ipassmt}.\ \ \textnormal{\texttt{sp}}\ \mathit{rs}\ \lbrace\rbrace \ \lbrace\rbrace
\end{IEEEeqnarray*}
\textnormal{then $\mathit{rs}$ provides spoofing protection according to Def.~\ref{def:nospoofpotentially}.}
\end{theorem}

The algorithm \texttt{sp}, presented in \mbox{Fig.\,\ref{algo:sp}}, is implemented recursively.
It iterates over the firewall ruleset.
\begin{figure*}
\begin{IEEEeqnarray*}{lcl}
	\textnormal{\texttt{sp}} \ [] \ A \ D & \ = \ & (A \setminus D) \subseteq \bigcup \mathit{ipassmt}[i] \\
    \textnormal{\texttt{sp}} \ ((m, \mathit{Accept}) \lstcons rs) \ A \ D & \ = \ & \textnormal{\texttt{sp}} \ rs \ (A \cup \lbrace \mathit{ip} \ \vert \ 
	   \exists \textnormal{$p$ from interface $i$ with src address $\mathit{ip}$}.\ \textnormal{matches} \ m \ p\rbrace ) \ D  \\
    \textnormal{\texttt{sp}} \ ((m, \mathit{Drop}) \lstcons rs) \ A \ D & \ = \ & \textnormal{\texttt{sp}} \ rs \ A \ \left(D \cup \left(\left\lbrace \mathit{ip} \ \vert \ 
       \forall \textnormal{$p$ from interface $i$ with src address $\mathit{ip}$}.\ \textnormal{matches} \ m \ p\right\rbrace \setminus A\right)\right) 
\end{IEEEeqnarray*}
  \caption{An algorithm to certify spoofing protection.}
  \label{algo:sp}
\end{figure*}

The base case is for an empty ruleset.
Here, $A$ and $D$ are the set of allowed/denied source IP addresses.
The firewall provides spoofing protection if the set of potentially allowed sources minus the set of definitely denied sources is a subset of the allowed IP range.

The two recursive calls collect these sets $A$ and $D$.
If the rule is an $\mathit{Accept}$ rule, the set $A$ is extended with the set of sources possibly accepted in this rule.
If the rule is a $\mathit{Drop}$ rule, the set $D$ is extended with the set of sources definitely denied in this rule, excluding any sources which were potentially accepted earlier.

As Theorem~\ref{thm:spsound} already states, \texttt{sp} can be started with any ruleset and the empty set for $A$ and $D$.

We will now describe how the $\mathit{ipaddr}\ \mathit{set}$ operations are implemented.
In general, we symbolically represent a set of IP addresses as a set of IP range intervals.
Since IP ranges are commonly expressed in CIDR notation (\eg $\mathit{a.b.c.d/n}$), the interval datatype proves to be very efficient.

Next, the set $\lbrace \mathit{ip} \ \vert \ \exists \textnormal{$p$ from interface $i$ with src address $\mathit{ip}$}.\allowbreak \textnormal{matches} \ m \ p\rbrace$ requires executable code.
Obviously, a straight-forward implementation which tests the existence of any packet is infeasible.
We provide an over-approximation for this set, the correctness proof confirms that this approach is sound.
First we check that $i$ matches all input interfaces specified in the match expression $m$. 
If this is not the case, the set is obviously empty.
Otherwise, we collect the intersection of all matches on source IPs in $m$.
If no source IPs are specified in $m$, then $m$ matches any source IP and we return the universe.

The set $\lbrace \mathit{ip} \ \vert \ \forall \textnormal{$p$ from interface $i$ with src address $\mathit{ip}$}.\allowbreak \textnormal{matches} \ m \ p\rbrace$ can be computed similarly.
However, we need to return an under-approximation here.
First, we check that $i$ matches, otherwise the set is empty.
Next, we remove all matches on input interfaces and source IPs from $m$.
If the remaining match expression is not unconditionally true, then we return the empty set.
Otherwise, we return the intersection of all source IP addresses specified in $m$, or the universe if $m$ does not restrict source IPs.

Note that after preprocessing we always have an explicit allow-all or deny-all rule at the end of the firewall ruleset.
Thus, $ A \cup D $ will always hold the universe after consuming the last rule.

\section{Evaluation -- Mathematically}
We outline the main idea of the correctness proof. 
Since \texttt{sp} operates on a fixed interface, we define certifiable spoofing protection for a fixed interface $i$.
Showing Def.~\ref{def:nospoofpotentiallyoneiface} for all interfaces is equivalent to Def.~\ref{def:nospoofpotentially}.

\begin{definition}
\label{def:nospoofpotentiallyoneiface}
\begin{IEEEeqnarray*}{c}
\forall p \in \left\lbrace \mathit{p} \ \vert \ \textnormal{$p$ from $i$ and potentially accepted by the firewall} \right\rbrace.\\
 \qquad $p$\textnormal{.src-ip} \in \bigcup \mathit{ipassmt}[i]
\end{IEEEeqnarray*}
\end{definition}

The correctness proof of \texttt{sp} is done by induction over the firewall ruleset.
Theorem~\ref{thm:spsound} does not lend itself to induction, since it features two empty sets which would generate unusable induction hypotheses.
To obtain a strong induction hypothesis, we generalize.
The ruleset is split into two parts: $\mathit{rs}_1$ and $\mathit{rs}_2$.
We assume that the algorithm correctly iterated over $\mathit{rs}_1$. 
For this lemma, we use the following notation:
\vskip-2.2em
\begin{IEEEeqnarray*}{l}
 A_\mathit{exact} = \lbrace \mathit{ip} \ \vert \ \exists p.\ \textnormal{$p$ from $i$ with src $\mathit{ip}$ and accepted by $\mathit{rs}_1$} \rbrace \\
 D_\mathit{exact} = \lbrace \mathit{ip} \ \vert \ \forall p.\ \textnormal{$p$ from $i$ with src $\mathit{ip}$ and denied by $\mathit{rs}_1$} \rbrace 
\end{IEEEeqnarray*}

\begin{lemma}
\label{lemma:nospoofgeneralized}
\textnormal{If} 
 $A_\mathit{exact} \subseteq A$ \textnormal{and} $D \subseteq D_\mathit{exact}$ \textnormal{and} $\textnormal{\texttt{sp}} \ \mathit{rs}_2 \ A \ D$ \textnormal{then Def.~\ref{def:nospoofpotentiallyoneiface} holds for} $\mathit{rs}_1 \lstapp \mathit{rs}_2$
\end{lemma}
\begin{proof}
The proof is done by induction over $\mathit{rs}_2$ for arbitrary $\mathit{rs}_1$, $A$, and $D$.

Base case (\ie $\mathit{rs}_2 = []$): From $\textnormal{\texttt{sp}} \ [] \ A \ D$ we conclude $(A \setminus D) \subseteq \bigcup \mathit{ipassmt}[i]$.
Since $A$ is an over-approximation and $D$ an under-approximation: $A_\mathit{exact} \setminus D_\mathit{exact} \subseteq A \setminus D$.
Since no IP address can be both accepted and denied we get $A_\mathit{exact} \setminus D_\mathit{exact} = A_\mathit{exact}$.
From transitivity we conclude $A_\mathit{exact} \subseteq \bigcup \mathit{ipassmt}[i]$, which implies spoofing protection for that interface according to Def.~\ref{def:nospoofpotentiallyoneiface}.

The two induction steps (one for $\mathit{Accept}$ and one for $\mathit{Drop}$ rules) follow from the induction hypothesis.
The over- and under- approximations were carefully constructed, such that the subset relations continue to hold.
The executable implementations of these sets also respect the subset relation; hence, the induction hypothesis solves these cases.
\end{proof}

\begin{proof}[Proof of Theorem~\ref{thm:spsound}]
Lemma~\ref{lemma:nospoofgeneralized} can be instantiated where $\mathit{rs}_1$ is the empty ruleset and $A$ and $D$ are the the empty set.
For this particular choice, it is easy to see that the preconditions hold.
Thus, for any $\mathit{rs}$, we conclude $\textnormal{\texttt{sp}} \ \mathit{rs} \ \lbrace\rbrace \ \lbrace\rbrace$ implies Def.~\ref{def:nospoofpotentiallyoneiface}.
Since this holds for arbitrary interfaces, we conclude Def.~\ref{def:nospoofpotentially}.
\end{proof}

Thus, our algorithm is proven \emph{sound} according to Def.~\ref{def:nospoofpotentially}.
This means, if the algorithm certifies a ruleset, then this ruleset is guaranteed to implement spoofing protection.

Note that our algorithm only certifies; debugging a ruleset in case of a certification failure remains manual. 
To debug, the proof of Lemma~\ref{lemma:nospoofgeneralized} suggest to consider the first rule where $(A \setminus D) \subseteq \bigcup \mathit{ipassmt}[i]$ is violated.

Standards such as Common Criteria~\cite{cc2012p3} require formal verification for their highest \emph{Evaluation Assurance Level} (EAL7), for example with Isabelle~\cite[\S{}A.5]{cc2012p3}. 
Therefore, any ruleset certified by Isabelle to provide spoofing protection could also be certified by Common Criteria EAL7.
Since we provide an executable algorithm, this can be done \emph{automatically} -- without the need for a user with formal background or manual proof.

The algorithm is \emph{not complete}.
This means, there may be rulesets which implement spoofing protection but cannot be certified by the algorithm.
This is bought by the approximations and by the support for unknown match conditions.
For example, the following ruleset cannot be certified:
\begin{Verbatim}[frame=single,label={firewall},labelposition=bottomline]
-i eth0 --src !192.168.0.0/24 --foo Drop
-i eth0 --src !192.168.0.0/24 --!foo Drop
any Accept
\end{Verbatim}

This is a reasonable decision for a completely unknown \verb~--foo~, since it might update an internal state and the mathematical equation ``$\verb~foo~ \vee \verb~!foo~ = \mathit{True}$'' may not hold.
However, if \verb~foo~ is replaced by the known and stateless match condition \verb~--protocol tcp~, the ruleset can be shown to correctly implement spoofing protection.
The algorithm, however, cannot certify it, since it does not track this match condition.
However, this is a made-up and bad-practice example, and we never encountered such special cases in any real-world ruleset.
The evaluation ---in which vast amounts of unknowns occurred--- shows that the algorithm certified all rulesets which included spoofing protection and correctly failed only for those rulesets which did not (correctly) implement it.
Thus, the incompleteness is primarily a theoretical limitation.

\section{Evaluation -- Empirically}
Often, firewalls start with an \verb~ESTABLISHED~ rule.
A packet can only match this rule if it belongs to a connection which has been accepted by the firewall previously.
Hence, the \verb~ESTABLISHED~ rule does not contribute to the access control policy for connection setup enforced by the firewall~\cite{fireman2006}.
Likewise, spoofed packets can only be allowed by the \verb~ESTABLISHED~ rule if they are allowed by any of the subsequent ACL rules. 
Therefore, as done in previous work~\cite[\S 6.4]{diekmann2015fm}, we either exclude this rule from our analysis or only consider packets of state \verb~NEW~.

We tested the algorithm on several real-word rulesets~\cite{diekmanngithubnetnetwork}. 
Most of them either did not provide spoofing protection or had an obvious spoofing protection and could thus be certified.
In this Section, we present the results of certifying the largest and most interesting ruleset.

First of all, for all rulesets, our algorithm was extremely fast:
Once the ruleset is preprocessed (few seconds) the certification algorithm only takes fractions of seconds for rulesets with several thousand rules.
We omit a detailed performance evaluation since these orders of magnitude are sufficient for a static/offline analysis system.\footnote{Certification runs of our algorithm were usually faster than reloading the ruleset on the firewall system itself.}

We present the certification of a firewall with about 4800 rules, connecting about 20 VLANs.
Every VLAN has its own interface.
Trying the certification, it immediately fails.
Responsible was a work-around rule which should only have existed temporarily but was forgotten.
This rule is now on the administrator's ``things to do the right way'' list and we exclude it for further evaluation.
Certifying spoofing protection for the first VLAN interface succeeds instantly.
However, trying to certify all other VLANs fails.
The reason is an error in the ruleset.
For every VLAN $n$, the firewall defines three custom chains: \texttt{mac\_}$n$, \texttt{ranges\_}$n$, and \texttt{filter\_}$n$.
The \texttt{mac\_}$n$ chain verifies that for hosts with registered MAC addresses and static IP addresses, nobody (with a different MAC address) steals the IP address.
This chain is primarily to avoid manual IP assignment errors.
Next, the \texttt{ranges\_}$n$ chain should prevent outgoing spoofing.
Finally, the \texttt{filter\_}$n$ chain allows packets for certain registered services.
The main error was that a spoofed packet from VLAN $m$ could be accepted by \texttt{filter\_}$n$ before it had to pass the \texttt{ranges\_}$m$ check.
The discovery of this error also discovered that the \texttt{mac\_}$m$ chains were not working reliably.
We verified these findings by sending and receiving such spoofed packet via the real firewall.
Finally, we fixed the firewall by moving all \texttt{mac\_}$n$ and \texttt{ranges\_}$n$ chains before any \texttt{filter\_}$n$ chains.
The certification for all but one\footnote{This one VLAN where the certification fails is only for internal testing purposes and deliberately features no spoofing protection} internal VLAN interfaces succeeds.
Next, the interfaces attached to the Internet are certified.
The IP address range was defined as the universe of all IPs, excluding the IPs owned by the institute.
Here, certification failed in a first run.
Responsible were some ssh rate limiting rules.
These rules were originally designed to prevent too many outgoing ssh connections.
However, since spoofing protection did not apply to them, an attacker could exploit them for a DOS attack against the internal network:
The attacker floods the firewall with ssh TCP SYN packets with spoofed internal addresses.
This exhausts the ssh limit of the internal hosts and it is no longer possible for them to establish new ssh connections.
This flaw was fixed and certification subsequently succeeds.
The improved and certified firewall ruleset is now in production use.

After this, another administrator got interested and wanted to implement spoofing protection for his firewall. 
To complicate matters, he was in Japan and the firewall in Germany. 
It was a key requirement that he would not lock himself out. 
Our tool could certify both: His proposed changes to the firewall correctly enforce spoofing protection and he will not lose ssh access. 
To provide a \emph{sound} guarantee for the latter, we applied the same idea as in our algorithm, but in reverse: the ruleset is abstracted to a stricter version (\ie a version that blocks more packets) and we consequently certify that it still allows \verb~NEW~ and \verb~ESTABLISHED~ ssh packets from the Internet.

\section{Conclusion}
We present an easy-to-use algorithm. 
It is fast enough to be run on every ruleset update.
It discovered real problems in a large, production-use firewall.
Both, the theoretical algorithm as well as the executable code are proven sound, hence if the algorithm certifies a firewall, the ruleset is \emph{proven} to implement spoofing protection correctly.
Sources and raw data: \cite{diekmanngithubiptablessemantics,diekmanngithubnetnetwork}.

\section*{Availability}
\label{sec:availability}
\noindent The analyzed firewall rulesets can be found at
\begin{center}
\vskip-0.8em
\url{https://github.com/diekmann/net-network}
\end{center}
\vskip-0.8em
\noindent
Our Isabelle formalization can be obtained from
\begin{center}
\vskip-0.8em
\url{https://github.com/diekmann/Iptables_Semantics}
\end{center}

\section*{Acknowledgments}
Andreas Korsten deserves our sincerest thanks for keeping our infrastructure running.
More than $20$ VLANs, transfer networks, an AS, several testbeds, Internet-wide scanning, hundreds of publicly accessible VMs randomly spawned on the fly, and probably thousands of special cases and exceptions.
This leads to an extremely complex infrastructure, all managed by one person; nevertheless he offers the time, commitment, and responsibility to deploy our newest research innovations on the live system.

Julius Michaelis implemented and proved large parts of the datatype to efficiently handle IP address ranges as intervals.

This work has been supported by the German Federal Ministry of Education, EUREKA project SASER, grant 16BP12304, and project SURF, grant 16KIS0145, and by the European Commission, project SafeCloud, grant 653884.


\bibliographystyle{IEEEtran}
\bibliography{IEEEabrv,../literature}

\end{document}